\documentclass[conference]{IEEEtran}

\usepackage[utf8]{inputenc}
\usepackage{amsmath,amssymb,bm}
\usepackage[dvips]{graphicx}
\usepackage{array}
\usepackage{cite}
\usepackage{enumerate}
\usepackage{color}
\usepackage{float}
\usepackage[capitalise]{cleveref}
\usepackage[mathcal]{euscript}
\usepackage{multirow,multicol,tabu}
\usepackage{subcaption}
\usepackage{soul}

\newcolumntype{?}{!{\vrule width 1pt}}

\usepackage{mathtools}	

\usepackage{algorithm}
\usepackage{algpseudocode}

\newcolumntype{M}[1]{>{\centering\arraybackslash}m{#1}}
\newcolumntype{N}{@{}m{0pt}@{}}

\usepackage{etoolbox}

\makeatletter
\newcommand{\changeoperator}[1]{%
  \csletcs{#1@saved}{#1@}%
  \csdef{#1@}{\changed@operator{#1}}%
}
\newcommand{\changed@operator}[1]{%
  \mathop{%
    \mathchoice{\textstyle\csuse{#1@saved}}
               {\csuse{#1@saved}}
               {\csuse{#1@saved}}
               {\csuse{#1@saved}}%
  }%
}
\makeatother

\changeoperator{sum}
\changeoperator{prod}

\usepackage{tikz,tabularx}
\usetikzlibrary{decorations.pathreplacing}
\tikzstyle{every picture}+=[remember picture]

\usetikzlibrary{shapes,arrows,shadows}	
\usepgflibrary{patterns}
\usetikzlibrary{patterns}
\usetikzlibrary{shapes.geometric}
\usetikzlibrary{arrows}

\usetikzlibrary{decorations.markings}

\tikzset{myptr1/.style={decoration={markings,mark=at position 1 with %
    {\arrow[scale=2.5]{>}}},postaction={decorate}}}
    
\tikzset{myptr2/.style={decoration={markings,mark=at position 1 with %
    {\arrow[scale=1.8]{>}}},postaction={decorate}}}
    
\usetikzlibrary{decorations.markings}

\usepackage{multicol}
\usepackage{amsthm}

\newtheorem{remark}{Remark}

\newtheorem{corollary}{Corollary}
\newtheorem{proposition}{Proposition}

\newcommand{\E}{\mathbb{E}}

\usepackage{subcaption}

\title{Download time analysis for distributed storage systems with node failures}

\author{
  	\IEEEauthorblockN{Tim Hellemans}
 	\IEEEauthorblockA{University of Antwerp, Belgium}
 	\IEEEauthorblockA{Tim.Hellemans@uantwerpen.be}		
 	\and
 	\IEEEauthorblockN{Arti Yardi}
 	\IEEEauthorblockA{IIIT Bangalore}
 	\IEEEauthorblockA{arti.yardi@gmail.com}		
 	\and
 	\IEEEauthorblockN{Tejas Bodas}
 	\IEEEauthorblockA{IIT Dharwad}
 	\IEEEauthorblockA{tejaspbodas@gmail.com}		
}

\begin{document}

\maketitle

\begin{abstract}
We consider a distributed storage system which stores several hot (popular) and cold (less popular) data files across multiple nodes or servers. Hot files are stored using repetition codes while cold files are stored using erasure codes. The nodes are prone to failure and hence at any given time, we assume that only a fraction of the nodes are available. Using a cavity process based mean field framework, we analyze the download time for users accessing hot or cold data in the presence of failed nodes. Our work also illustrates the impact of the choice of the storage code on the download time performance of users in the system. 
\end{abstract}

%
\section{Introduction}
\label{Section_Introduction}
The use of distributed storage systems (DSS) has become widespread due to the large amount of data that is being generated everyday. The idea behind DSS is to store data across multiple data centers in a distributed manner. To ensure that data is not lost due to node failures, coding theoretic techniques are used to introduce redundancy across stored data \cite{Dimakis11,Balaji18}. In its simplest form, replication codes are very effective in preventing data loss since multiple copies of the same file are stored across different nodes. Due to its simplicity, replication codes are used to store hot data, i.e., data or content that is very popular \cite{Ghemawat03,Shvachko10,Aktas20}. Recent techniques for data storage use erasure codes such as maximum distance separable (MDS) codes \cite{Dimakis10}, regenerating (RG) codes \cite{Dimakis10}, minimum storage regenerating (MSR) codes, minimum bandwidth regeneration (MBR) codes \cite{Rashmi09, Rashmi11}, and locally recoverable (LR) codes \cite{pp14}. In such codes, a file of size $B$ is divided into $k$ equal parts which are then encoded into $d \geq k$ parts each of size {$\theta$ with $\theta \geq \frac{B}{k}$ (the total storage needed is $ d \theta  > B$}). The original file can be reconstructed after downloading any $k$ out of these $d$ parts. Compared to repetition codes, erasure codes are effective in introducing redundancy and are robust to the problem of node failures \cite{Huang2012}. However, reconstructing a coded file from its fragments involves overheads and hence erasure codes are only used to store less popular content. 

Node failure is a major issue plaguing DSS. When a node fails, a repair is initiated to recover the lost data and this recovery traffic is non-negligible and can even impact the download time for hot and cold data. Therefore node recovery is only performed periodically at a time when the network activity is at its minimum. This leaves a lot of nodes unavailable for use during peak activity and this impacts the download time for file request. A clear understanding of this impact is missing in the literature. When it comes to analyzing the download time, an important open problem is the exact download time analysis for erasure coded data. While the literature predominantly offers bounds on the mean download time or on the tail latency, the analysis is also under simplifying assumptions such as unique file to download, exponential  service times, absence of node failures etc. For example, in one of the earliest work on MDS codes, Shah et al \cite{Shah16} provide bounds on the mean download time for a single erasure coded file stored across all its servers. \cite{Joshi12,Joshi15} consider the trade-off between storage space and download time of a DSS and show that the download time can be reduced by using erasure codes over replication. \cite{Li18} demonstrates the latency gains with respect to using erasure coding over replication in DSS.  \cite{Aggarwal17,Al2019} provide bounds on the tail latency of a distributed storage system with multiple files under probabilistic scheduling of jobs and general service times. \cite{Badita19} analyzes the MDS queue and offers yet another bound on the mean download time for erasure codes.
The download time analysis for repetition codes is relatively easy and closed form expressions for the mean response time have been obtained recently in \cite{Ayesta18,Ayesta19,Ayesta19b} for exponential service times. For general service times, a mean field analysis is presented in \cite{Hellemans18}. 

In this work we consider the download time analysis in a DSS with $N$ nodes that store several hot and cold files and where only a fraction $q$ of the nodes are available.
Nodes follow a first in first out (FIFO) scheduling policy to serve requests and the service time variable has a general distribution. Files of type~$i$ are downloaded using the least loaded $k_i$ out of $d_i$ policy (LL($d_i,k_i$)) policy where $i \in \{h,c\}$ for hot and cold files. In this policy, each type $i$ component is requested from $d_i$ appropriate  servers that store it. The first $k_i$ file components that start being downloaded are retained and the rest are canceled.
For hot files, $i=h$ and $k_h = 1$. 
{The novelty of this paper is in its ability to model a DSS with node failures, that stores several hot and cold files and where the random service requirement follows a shifted phase-type distribution. To accommodate all these features, we model the system as a workload dependent load balancing policy and employ the cavity process based mean field method for its analysis. We believe that this method is a very powerful tool that is able to model and analyze such intricate system details accurately.
To the best of our knowledge, no recent or previous work has been able to analyze the download time for erasure coded data with all these features at once.
}

 The use of mean field approximation for analyzing load balancing policies is quite popular \cite{Mitzenmacher01,Bramson10,Mukherjee16,Hellemans18,Jinan20}. Our download time approximation is based on mean field asymptotics ($N \rightarrow \infty$) for the workload distribution at an arbitrary queue. Here, the workload at a node represents the actual pending requests that will be served at the node (ignoring the pending requests at the node that will be canceled due to the $LL(d,k)$ policy employed). 
 Based on the cavity queue approach used in \cite{Bramson10,Bramson11,Bramson12,Hellemans18,Hellemans19}, and the asymptotic independence of the node workloads under the $LL(d,k)$ policy \cite{Shneer20}, we obtain an integro-differential equation (IDE) characterizing the workload distribution at an arbitrary queue (the cavity queue). This enables us to obtain the download time distribution for hot and cold file requests.
 
 \subsubsection*{Organization} In Section \ref{sec:SystemModel} we describe the system model along with some preliminaries. Our main results on the workload and download time distribution are presented in Section~\ref{Section_results}. Numerical results are presented in Section~\ref{Section_simulations} followed by a conclusion in Section~\ref{Section_Conclusion}. 
 
%
%
%
%
%

%
\section{System Model and Preliminaries }
\label{sec:SystemModel}

In this work, we perceive a DSS as a load balancing system with $N$ data storage servers. Requests for file downloads arrive at a dispatcher which forwards the requests to appropriate nodes that store the file. Each server has a FIFO queue for placing the file requests and hence the download time of a file not only depends on its file size (or service time) but  also on other file requests that have arrived before it. We assume that file requests arrive according to a Poisson process with rate $\lambda N$. Further the requested file is of type $i$ with probability $p_i$ for $i \in \{h,c\}$. 
For ease of exposition, we assume that each file (irrespective of its type) has size $B$ and each server has the same storage capacity.
Each type $i$ file is stored across $d_i$ distinct servers using a $(d_i,k_i)-$ erasure code $C_i$ of dimension $k_i$ for $i \in \{h,c\}$. 
In such a coding technique, a file is encoded into $d_i$ \textit{fragments} (to be stored across $d_i$ servers) and any $k_i$ out of the $d_i$ fragments are sufficient to reconstruct the file. We assume that each fragment belongs to some finite field $\mathbb{F}_{\hat{q}}$ and the size of each coded fragment for a type $i$ file is $\theta_i$. When the underlying code $C_c$ is an MDS code,
we have $\theta_i  = \frac{B}{k_i}$.  When the underlying code $C_c$ is an RG code, each node stores $\theta_i \geq \frac{B}{k_i}$ symbols from $\mathbb{F}_{\hat{q}}$. An erasure coded file of size $B$ therefore requires a storage space of  $d_i \theta_i$ and can be reconstructed after downloading $k_i \theta_i \geq B$ units.
We use repetition code $C_h$ for hot files and hence $k_h = 1$ and $\theta_h = B$.

{A requirement of the cavity process method is that different servers must be sampled at an identical rate and hence should have identical load. If this is not the case, we would have to use multiple cavity queues to represent differently loaded servers, thus complicating the analysis. To alleviate these difficulties, we make the following assumptions.}
While storing a type $i$ file (this happens only once and at the time of setting up the DSS), we assume that the corresponding $d_i$ servers are chosen uniformly at random, $i \in \{h,c\}$. We further assume that two files of the same type are not stored on the same set of $d_i$ servers. For hot files, this ensures that at most one hot file is lost from any $d_i$ node failures. Under such storage restrictions, our DSS can store ${N \choose d_i}$ distinct type $i$ files across $N$ nodes. This allows us to map a set of $d_i$ servers to a unique type $i$ file. We assume that files have identical popularity and hence each file is requested uniformly at random. With these assumptions, we can now view the DSS as a power-of-$d_i$ choice load balancing system where an arriving request of type $i$ samples $d_i$ servers uniformly at random \cite{Ayesta18, Hellemans19}. See \cite{Shah17} for a similar file placement policy that ensure symmetry across storage nodes.


 We assume that a type $i$ file is downloaded according to the $LL(d_i,k_i)$ policy (LL stands for Least Loaded).  When a request for a type $i$ file arrives, the dispatcher copies this request on the corresponding $d_i$ servers. For each file, the first $k_i$ out of $d_i$ copies to reach the head of their respective queue are retained while the remaining $d_i - k_i$ copies are deleted. The $k_i$ components are downloaded from the least loaded $k_i$ out of $d_i$ servers and are sufficient to reconstruct the file. Here the load (or workload) at a server as seen by an arriving request is the pending work at the server that will be served before this  request starts receiving any service. See \cite{Hellemans19} for more details of the policy.

 For a type $i$ file, we define the \textit{fragment service time} $\hat{X}_i$ associated with a file fragment as the time taken to serve the particular fragment request once it is picked by a server for service (actual download). We assume that $\hat{X}_i = \delta_i + Y_i$ where $\delta_i$ is a non-negative constant startup time taken by the server while serving the fragment. $Y_i$ is a PH($\alpha_i, A_i$) (phase-type) random variable where $\alpha_i \in \mathbb{R}^n$ and matrix $A_i \in \mathbb{R}^{n \times n}$ for $n \in \mathbb{N}$. Defining $\mu_i = -A_i \textbf{1}$, the probability density function (pdf) $f_{Y_i}(w) = \alpha_i e^{A_iw}\mu_i$.  We denote the cumulative distribution function (CDF) of $\hat{X}_i$ by $G_i(\cdot)$, its complimentary CDF (CCDF) by $\bar{G}_i(\cdot)$, its pdf by $g_i(\cdot)$ and its mean by $E[\hat{X}_i]$.
 {Note that we have modeled $\hat{X}_i$ by a shifted-phase type distribution. The shifted exponential distribution that is used to model jobs in distributed computing \cite{Peng20} is a special case of this distribution. We use phase type distribution because they are known to be dense in the class of all distributions and hence any distribution may be approximated by it.} 
We also assume that for a file request, the corresponding $\hat{X}_i$ for each of the $k_i$ fragments  is identical. We denote the \textit{file service time} for a type $i$ file by $X_i$ where due to identical fragment size assumption, we have $X_i = k_i \hat{X}_i$ with mean $E[X_i]$. A file service time models the random time taken to retrieve $k_i \theta_i$ amount of data directly from a single server. 
 We use the term \textit{fragment download time} or \textit{response time} to indicate the total time spent by a fragment in the system (sum of the  fetching time and its waiting time).  The \textit{download time} or \textit{response time} $R_i$ for a type $i$ file then denotes the time taken by the file request since its arrival to finish downloading all the $k_i$ components. As part of our main result, we characterize the CCDF of $R_i$ using the mean field approach.  

\subsubsection*{The cavity process method}
This is a mean field approximation method used to obtain the workload distribution $F(\cdot)$ at an arbitrary queue of our DSS when $N$ goes to infinity. In order to employ this methodology, the asymptotic independence of the workload at different queues must be established. Fortunately for load balancing policies that are combinations of $LL(d_j,k_j)$ (which form the basis of our paper), this was proven by Shneer and Stoylar \cite{Shneer20}. The asymptotic independence ensures that the workload at the different queues are independent and governed by the same law $F(\cdot).$ We will denote the workload density by $f(\cdot)$ and its CCDF by $\bar{F}(\cdot)$. The main task is to characterize  $\bar{F}(\cdot)$ for an arbitrary queue in our DSS. Following an approach similar to that in \cite{Bramson11,Hellemans18,Hellemans19} we consider a tagged queue (or a cavity queue) and analyze it as an $M/G/1$ queue with workload CCDF $\bar{F}(\cdot)$ (this is yet to be identified) and workload dependent arrival rate $\lambda(w)$ (this can be determined for the $LL(d_i,k_i)$ policy). Without divulging too much details (c.f. \cite{Bramson10,Hellemans18}), the high level idea involved in obtaining $\bar{F}(\cdot)$ comprises the following two steps. In the first step, assuming that $\bar{F}(\cdot)$ is given, one uses the properties of $LL(d_i,k_i)$  to obtain the arrival rate $\lambda(w)$ to a cavity queue with workload $w$. We denote this step by the expression $\lambda(w) = H(\bar{F}(w))$. As part of the second step, for a given workload dependent arrival rate $\lambda(w)$ to the cavity queue, we use properties of a standard M/G/1 queue to obtain $\bar{F}(w)$ for $w \geq 0$. This step is denoted as $\bar{F}(w) = \Phi(\lambda(w))$. The two steps are combined to obtain a functional fixed point equation for the stationary workload distribution $\Phi (H(\bar{F}(w))) = \bar{F}(w)$.  See \cite{Hellemans19} for details on obtaining the  functional differential equation for $\bar{F}(\cdot)$ in a system with a vanilla $LL(d,k)$ policy. At a high level, our DSS model with node failures can be seen as a combination of the LL($d,k$) policy for different choices of $d$ and $k$. However the application of the cavity process method to obtain $\bar{F}(\cdot)$ for our DSS is straightforward (see below). 
%
\section{Main Results}
\label{Section_results}
In order to analyze the impact of node failures, we assume that a randomly sampled node is available (not failed) with probability $q$. Failed nodes affect the data availability which in turn affects the download time. For example, for a cold file request, it may happen that less than $k_c$ out of the $d_c$ nodes storing the required fragments are available. This not only leads to a loss of the file request but also wastes server resources due to download of available fragments (which are not sufficient to reconstruct the file). We first define the quantity $B_j(d_i)$ as the probability that exactly $j$ out of $d_i$ selected servers are working (and $d_i-j$ servers have failed).
It is easy to see that the probability that an arriving type $i$ request is lost, denoted by $p_{l,i}$ is given by:
\begin{equation} \label{eq:loss_prob}
p_{l,i} = \sum_{j=0}^{k_i-1} B_j(d_i)
\end{equation}
where $i \in \{h,c\}$ and $B_j(d_i) = \binom{d_i}{j} q^j \cdot (1-q)^{d_i-j}$. Note that $B_j(d_i)$ plays a role in characterizing the load at a server. For our system, we find that the {system load} (defined as the proportion of servers that are busy) is given by:
\begin{eqnarray}\label{eq:load_model1}
\rho &=&  \lambda p_h \E[\hat{X}_h] (1- B_0(d_h)) \nonumber \\ 
&+& \lambda p_c \E[\hat{X}_c] \cdot \left(\sum_{i \geq 1} \frac{\min\{i, k_c\}}{k_c} B_i(d_c) \right). 
\end{eqnarray}
The right hand side can be interpreted as the amount of incoming work per unit time. As servers work at unit rate  when they have some work, the conservation of work  implies that, at equilibrium, the probability that a server is busy must equal the amount of incoming work per unit time.

\subsubsection*{Workload distribution}
To obtain an Integro Differential Equation (IDE) in terms of $\bar{F}(\cdot)$, we distinguish between potential and actual arrivals at the cavity queue. We say a potential arrival of type $i$ occurs at the cavity queue whenever it is selected by a file request for service. As arrivals of type $i$ occur with a rate $p_i \lambda N$, and any queue has a probability of $\frac{d_i}{N}$ of being selected, the potential arrival rate of type $i$ files is $p_i \lambda d_i$.
The main idea of the following proof is to note that the LL($d, k$) policy in the presence of failed servers reduces to a convex combination of LL($j, \min\{k, j\}$) type policies. Indeed, when only $j$ of the $d$ selected servers are available (say $j \geq k$), we try to finish the $k$ fragments on these $j$ servers.

\begin{proposition} \label{prop:model1}
The equilibrium workload distribution for our DSS model satisfies the following IDE:
\begin{align*}
\footnotesize
\bar F'(w) &= -\sum_{i \in \{h,c\}} \lambda p_i \sum_{j=1}^{d_i} B_j(d_i) \left[ \bar G_i(w) + H_{j, \min{(j,k_i)}}(w) \right. \\ &- \left. \int_0^w H_{j,\min{(j,k_i)}}(u) g_i(w-u) \, du \right]
\normalsize
\end{align*}
with $\bar F(0) = \rho$ (as in \eqref{eq:load_model1}) and:
\begin{equation}\label{eq:HdK}
H_{j,k}(w) = \sum_{i=1}^j \min\{i, k\} \binom{j}{i} \bar F(w)^{j-i} (1-\bar F(w))^i - 1.
\end{equation}
\end{proposition}
\begin{proof}
Assume we have some policy $P$ which applies policy $P_1$ with probability $p_1$, and $P_2$ with probability $p_2$. Moreover, denote by $d_i$ the number of servers selected for policy $P_i$.
Let $U$ denote the workload at the cavity queue right before a potential arrival, $\mathcal{Q}(U)$ the workload at the cavity queue right after the potential arrival and $\mathcal{Q}_j(U)$ the workload right after a potential arrival under policy $P_j$. Note that the potential arrival is an actual arrival if and only if $\mathcal{Q}(U) > U$. From Theorem 5.2 in \cite{Hellemans19}, one can see that $\bar F'(w)$ satisfies
\begin{align}
\bar F'(w) &= -\lambda d_1 p_1 \mathbb{P}\{ \mathcal{Q}_1(U) > w , U \leq w \}\nonumber \\
& - \lambda d_2 p_2 \mathbb{P}\{ \mathcal{Q}_2(U) > w, U \leq w \}. \label{eq:FDE_linear}
\end{align}
For our DSS with hot and cold files, this implies that
\begin{align}
\bar F'(w) &= -\lambda d_c p_c \mathbb{P}\{ \mathcal{Q}_c(U) > w , U \leq w \}\nonumber \\
& - \lambda d_h p_h \mathbb{P}\{ \mathcal{Q}_h(U) > w, U \leq w \}, \label{eq:FDE_ch}
\end{align}
with $\mathcal{Q}_c(U)$ and $\mathcal{Q}_h(U)$  denotes the workload at the cavity queue right after a potential cold and hot arrival.
Let us now consider the case when a cold file request finds exactly $j$ functioning servers out of $d_c$. This happens with probability $B_j(d_c)$ and the file is served using the policy\ LL($j, \min\{k_c, j\}$).
For a hot file  that finds $j$ out of $d_h$ servers available, the file is served using LL($j,1$). Treating each such case as a separate policy and  using \eqref{eq:FDE_linear}, we can  now write the right hand side of \eqref{eq:FDE_ch} as:
\begin{align}
&-\lambda d_c p_c \sum_{j=1}^{d_c} B_j(d_c) \mathbb{P}\{ \mathcal{Q}_{LL(j, \min\{k_c, j\}) }(U) > w , U \leq w \} \nonumber\\
& - \lambda d_h p_h \sum_{j=1}^{d_h} {B_j(d_h)} \mathbb{P}\{ \mathcal{Q}_{LL(j,1)}(U) > w, U \leq w \}, \nonumber
\end{align}
The proof now follows from Proposition 6.8 in \cite{Hellemans19}, where 
$
\mathbb{P}\{ \mathcal{Q}_{LL(d, k) }(U) > w , U \leq w \}
$
is computed.
\end{proof}
\begin{remark}
From this result one easily finds the mean workload $\E[W] = \int_0^\infty \bar F(w) \, dw$ for the cavity queue.
\end{remark}
\begin{remark}
There are two special cases for $H_{j,k}(w)$ as defined in \eqref{eq:HdK}. The first is the case where $k=1$, in which case $H_{j,1}(w) = -\bar F(w)^j$ (as is the case for hot jobs). The second special case is $k_i = d_i$ in which case we have $H_{j,j}(w) = j-1 - j \bar{F}(w)$ (as is the case for cold jobs which find $j \leq k_c$ functioning servers).
\end{remark}
Using the fact that $Y_i$ is a PH($\alpha_i, A_i$) variable, we can instead obtain $\bar F(w)$ as the solution of a delayed differential equation (DDE). Note that solving a DDE requires only O($K$) computations while the IDE requires O($K^2$) ($K$ denotes the number of points used to discretize the workload).
\begin{corollary}\label{cor:model1}
When $\hat{X}_i = \delta_i + Y_i$ with $\delta_i \in [0,\infty)$ and $Y_i$ has distribution PH($\alpha_i, A_i$) for $i \in \{h,c\}$, the equilibrium workload distribution satisfies the following DDE: 
\begin{align*}
\footnotesize
\bar F'(w) &= - \sum_{i \in \{h,c\}} \lambda p_i \sum_{j=1}^{d_i} B_j(d_i) \\ & \times \left[ \bar G_i(w) + H_{j,\min\{j, k_i\}}(w) - \alpha_i \xi_{j,\min\{j, k_i\},\delta_i}^{(i)}(w) \right]
\end{align*}
\normalsize
with $\xi^{(i)}_{.,.,\delta_i}(w)= 0$ (for $w \leq \delta_i$, $\bar F(0) = \rho$ (as in \eqref{eq:load_model1}) and:
\begin{align*}
\xi_{j,k,\delta_i}^{\prime(i)} (w) &= A_i \xi_{j,k,\delta_i}^{(i)} (w) + H_{j,k}(w-\delta_i)^{d_i} \mu_i, \quad w > \delta_i.
\end{align*}
\end{corollary}
\begin{proof}
The proof follows by defining:
$$
\xi_{j,k,\delta_i}^{(i)} (w) = \int_0^{w-\delta_i} H_{j,k}(u) e^{A_i(w-u-\delta_i)} \mu_i \, du, w > \delta_i,
$$
and $\xi^{(i)}_{.,.,\delta_i}(w)= 0$ for $w \leq \delta_i$.
\end{proof}


\begin{figure*}[h!]
\begin{center}
\begin {tikzpicture}
%
    %
    \node[inner sep=0pt] (whitehead) at (0,0)  
    {\includegraphics[width=6.5cm]{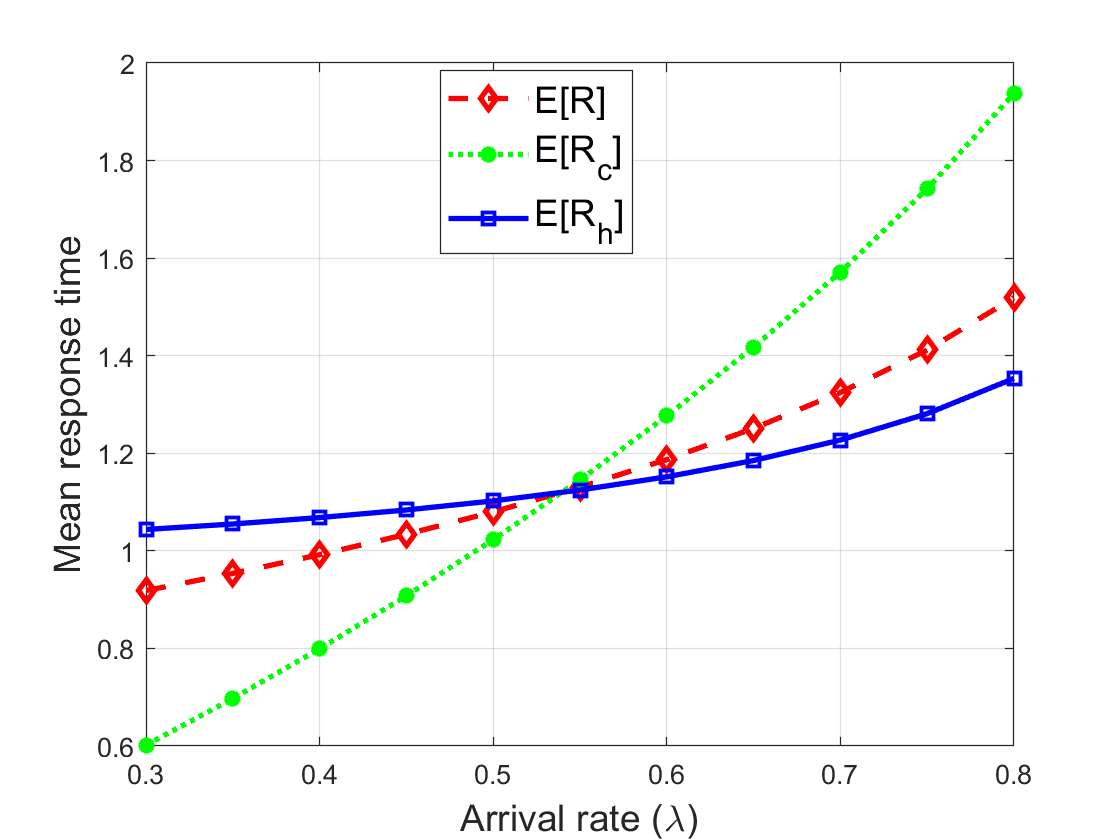}};
     \node [above] at (0,-3.2) {\scriptsize (a)};     
    \node[inner sep=0pt] (whitehead) at (6,0)  
    {\includegraphics[width=6.5cm]{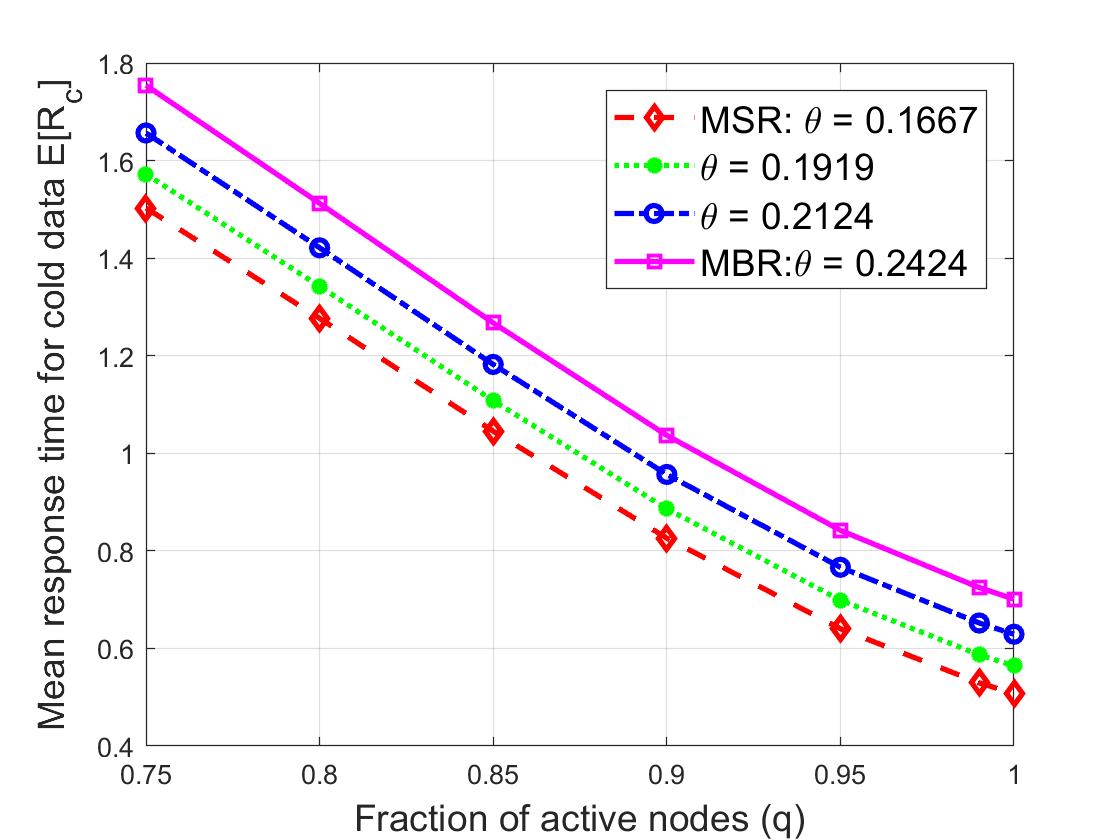}};
    \node [above] at (6,-3.2) {\scriptsize (b) };     
    %
    %
    \node[inner sep=0pt] (whitehead) at (12,0)  
    {\includegraphics[width=6.5cm]{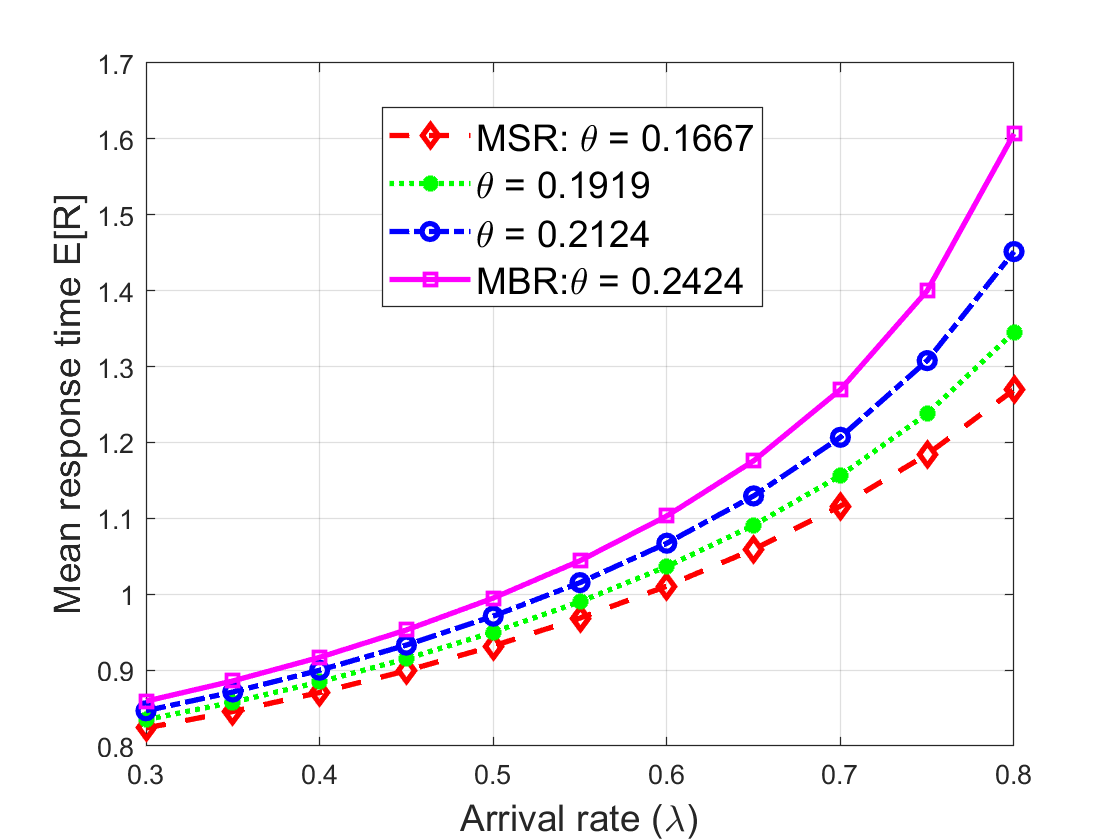}};
    \node [above] at (12,-3.2) {\scriptsize (c)};     
    %
    %
\end {tikzpicture}
\caption{Plot of mean response time for various parameters}
\vspace{-0.2in}
\label{Figure_ER_curves}
\end{center}
\end{figure*}
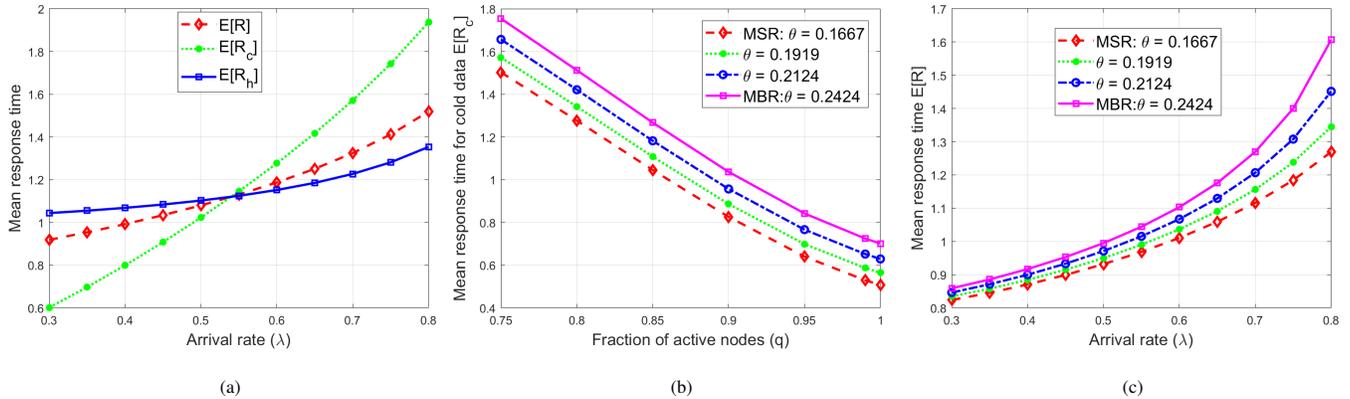

\subsubsection*{Download time distribution}
We now characterize the download time or response time distribution for hot and cold requests. Note that a type $i$ file request can be successfully downloaded only if at least $k_i$ of the $d_i$ nodes storing the file have not failed.
Now consider such a type $i$ file request which has exactly $j$
of its $d_i$ nodes available where $j \geq k_i$.
Let $\{U_n, n= 1, \ldots, j\}$ denote the workload variable (pending work) at each of the $j$ available servers. Since the policy employed is $LL(d_i, k_i)$ and the fragment service time is identical, the file is downloaded once the file request at the $k_i^{th}$ least loaded queue reaches the head of the queue and the corresponding fragment is served. Therefore the download time variable $R_i$ for this particular type $i$ file (that has exactly $j$ nodes available) is equal to $U_{(j, k_i)} + \hat{X}_i$ where $U_{(j, k_i)}$ denotes the $k_i^{th}$ order statistic of $\{U_n, n= 1, \ldots, j\}$ where $k_i \leq j \leq d_i$. Due to the asymptotic independence of the queues, $\{U_n, n= 1, \ldots, j\}$ are i.i.d random variables with CDF $F(\cdot)$ and hence  $U_{(j, k_i)}$ now denotes the $k_i^{th}$ order statistic of $j$ i.i.d random variables. 
Let $R_{(j,k_i)}$ denote the download time variable for a type~$i$ file request for which exactly $j$ out of $d_i$ nodes are available and $j \geq k_i$. We find that $R_{(j, k_i)} = U_{(j, k_i)} + \hat{X}_i$ and hence its CCDF is given by:
\begin{align}
\label{eq:frjk}
\bar F_{R_{(j,k_i)}}(w)
&= \bar G_i(w) + \int_0^w \bar F_{U_{(j,k_i)}}(u) \cdot g_i(w-u) \, du.
\end{align}
Now if the $k_i^{th}$ order statistic satisfies  $U_{(j,k_i)} > u,$ then this implies that at least $j - k_i + 1$ nodes have workload greater than $u$. Therefore:
\begin{equation*}
 \bar{F}_{U_{(j,k_i)}}(u) = \sum_{i=0}^{k_i} \binom{j}{i} \bar{F}^{j-i}(u){F}^{i}(u).
 \end{equation*}
We use the preceding discussion to characterize the conditional download time or response time distribution $F_{R_i}$ for an arriving type~$i$ job conditioned on the fact that it finds at least $k_i$ out of $d_i$ nodes available (otherwise the file cannot be reconstructed). Let  $\tilde p^{(i)}_j, j \geq k_i$ denote the probability that an arriving job is of type $i$ and it finds exactly $j$ out of the $d_i$ nodes storing this file to be available.  Clearly we have
$
\tilde p^{(i)}_j = p_i  B_j(d_i)
$
for $i \in \{h,c\}$.
To account for only those file requests that can be reconstructed, we normalize $\tilde p^{(i)}_j$ by dividing it with $p_i \sum_{j=k_i}^{d_i}  B_j(d_i)$ for each $j$. The resulting quantity denoted by $p^{(i)}_j$ is the probability that an arriving request for a type $i$ file finds $j$ out of its $d_i$ servers available given that at least $k_i$ out of $d_i$ files storing it are available.
From this we are able to compute the response time distribution for a type~$i$ job as:
\begin{equation}
\label{eq:resp_time}
 \bar F_{R_i}(w) = \sum_{j \geq k_i} p^{(i)}_j \bar F_{R_{(j,k_i)}}(w),
\end{equation}
where $\bar F_{R_{(j,k_i)}}(w)$ is given by \eqref{eq:frjk}.
Integrating these functions yields the mean response time for cold and hot jobs given by $\E[R_i]=\int_0^\infty \bar F_{R_i}(w)\, dw$. In order to compute the response time distribution of an arbitrary job, we need to compute the probability that an arbitrary job which is finished is a hot/cold job. To this end we denote $\tilde \beta_i = \sum_{j \geq k_i} \tilde p_j^{(i)}$ and then re-normalize these to have $\beta_i,$ such that $\beta_h + \beta_c = 1$. We find that the CCDF of the response time is given by:
$
\bar F_R(w) = \beta_h \bar F_{R_h}(w) + \beta_c \bar F_{R_c}(w),
$
integrating again yields the mean response time.


%
\section{Numerical results}
\label{Section_simulations}

In this section, we offer numerical insights into the impact of node failures on $E[R_i]$ for different types of erasure codes. In all our results, the numerical evaluation for $E[R_i]$ first involves evaluating $\bar{F}(\cdot)$ using Corollary~\ref{cor:model1}. This is then used for the numerical evaluation of $\bar{F}_{R_i}(\cdot)$ in \eqref{eq:resp_time} which is then integrated to obtain $E[R_i]$. Needless to say, our results are for the limiting regime when $N \rightarrow \infty$. Having said that, the results are very accurate for finite but large values of $N$. We do not exhibit such results here for space constraint but refer readers to \cite{Hellemans18,Hellemans19} for similar observations.
\subsubsection*{Example~1}
We first consider a DSS with the following set of parameters. The arrival rate is $\lambda = 0.7$. In practice, around 70 percent of the download traffic comprises of hot files and hence we choose $p_h = 0.7$. Hot files are stored using 3-replication $(d_h = 3)$ and have a mean \textit{file service time} $E[X_h] = 1$. Cold files are stored using a (4,2)-MDS code ($d_c = 4$ and $k_c = 2$) with 
mean \textit{file service time} also satisfying $ E[X_c] = 1$. 
We assume that the corresponding file service time distribution is hyper-exponential with two phases and with parameters adjusted to have squared coefficient of variation 
$SCV = 2$ and shape parameter $f = 0.5$ (balanced means for the phases). See \cite{Hellemans18} for details on relating SCV and $f$ of a hyper-exponential with the underlying phase-type parameters. Finally we assume that the constant startup time for serving a fragment is $\delta_h = \delta_c = 0.2 $. For a DSS with these parameters, \cref{table:ex1} compares $E[R_i]$ for type $i$ files for various values of $q$ (fraction of available nodes). As the table suggests, as $q$ decreases, $E[R_i]$ increases and so does the loss probabilities. To illustrate the impact of node failure, the table outlines the percentage increase in $E[R_i]$ compared with the $q=1$ case.   


\begin{table}[htbp]

\begin{center}

\begin{tabular}{|c|c|c|c|c|c|c|c|}
\hline
  \multirow{3}{*}{$q$} 
& \multicolumn{2}{c|}{Percentage increase} 
& \multicolumn{2}{c|}{Loss probability} \\
& \multicolumn{2}{c|}{in response time} 
& \multicolumn{2}{c|}{~} \\
\cline{2-5}
& Hot file & Cold file  
& Hot file & Cold file \\
\hline
%
%
$1$    & $0\%$     & $0\%$     & $0$      &  $0$ \\ \hline
$0.95$ & $2.3\%$   & $8.08\%$  & $0.0001$ &  $0.005$ \\ \hline
$0.9$  & $5.16\%$  & $18.01\%$ & $0.001$  &  $0.0037$ \\ \hline
$0.85$ & $8.51\%$  & $29.2\%$  & $0.0034$ &  $0.0120$ \\ \hline
$0.8$  & $12.24\%$ & $41.13\%$ & $0.008$  &  $0.0272$ \\ \hline
$0.75$ & $16.21\%$ & $53.31\%$ & $0.0156$ &  $0.0508$ \\ \hline
$0.7$  & $20.26\%$ & $65.23\%$ & $0.027$  &  $0.0837$ \\ \hline
$0.65$ & $24.19\%$ & $76.32\%$ & $0.0429$ &  $0.1265$ \\ \hline
$0.6$  & $27.78\%$ & $85.98\%$ & $0.064$  &  $0.1792$ \\ \hline
\end{tabular}

\end{center}

\caption{Percentage increase in $E[R_i]$ versus $q$}
\label{table:ex1}
%
\end{table}
\subsubsection*{Example 2}
We consider a (9,6)-MDS code
with $d_c = 9, k_c = 6$. The file service time distribution is again a hyper-exponential distribution with $E[X_c] = 1,
SCV = 2$ and $f = 0.5$. For hot files, we have $p_h = 0.7, d_h = 3, E[X_h] = 1,$ and the constant startup time 
$\delta_h = 0.1$ and $\delta_c = 0.1$. Assuming that the fraction of available nodes is $0.8,$ we plot $E[R], E[R_h]$ and $E[R_c]$ versus $\lambda$ for our DSS in \cref{Figure_ER_curves}-(a). While the response time increases with $\lambda,$ it is interesting to note that the gains due to download parallelism involved in cold files decreases with $\lambda$.

\subsubsection*{Example 3}
In \cref{Figure_ER_curves}-(b) and \cref{Figure_ER_curves}-(c) we plot $E[R_c]$ versus $q$ and $\lambda$ respectively for (9,6)-RG codes with different fragment sizes $\theta$. The corresponding file service time is set as $E[X_c] = k_c \theta$. The remaining parameters are chosen as in Example 2. For this example, we assume a normalized file size of $B=1$. Since $d_c = 9$ and $k_c = 6,$ the MSR code variant has $\theta = \frac{1}{6}$.
These codes have the lowest value of $\theta$ and hence have the lowest $E[R_c]$ across the two plots. The other extreme are the MBR codes with $\theta = 0.2424$. Since their mean file service time is highest of the 4 codes we consider, they have the highest value of $E[R_c]$. In both the plots, $E[R_c]$ decreases with $q$ and increases with $\lambda$.

%
\section{Conclusion}
\label{Section_Conclusion}

In this work, we have used the cavity process method to characterize the mean download time for a DSS with hot and cold files and node failures. We obtain the DDE for the workload CCDF which can be easily evaluated numerically.
This is then used to obtain the mean download time for hot and cold jobs. A clear advantage of the mean field approach is that we are able to model node failures and also account for multiple hot and cold files, something which has not been possible to analyze till now. 
Our method also allows us to see the impact of erasure codes and their parameters on the mean download time. As future work, we would like to account for recovery traffic and their impact on the download time of hot and cold files. Further, it would be interesting to have an accurate model to account for node failures where the nodes resume service as soon as they have been recovered. 


\bibliographystyle{IEEEtran}
\bibliography{References_Paper1_Tim}

\begin{thebibliography}{10}
\providecommand{\url}[1]{#1}
\csname url@samestyle\endcsname
\providecommand{\newblock}{\relax}
\providecommand{\bibinfo}[2]{#2}
\providecommand{\BIBentrySTDinterwordspacing}{\spaceskip=0pt\relax}
\providecommand{\BIBentryALTinterwordstretchfactor}{4}
\providecommand{\BIBentryALTinterwordspacing}{\spaceskip=\fontdimen2\font plus
\BIBentryALTinterwordstretchfactor\fontdimen3\font minus
  \fontdimen4\font\relax}
\providecommand{\BIBforeignlanguage}[2]{{%
\expandafter\ifx\csname l@#1\endcsname\relax
\typeout{** WARNING: IEEEtran.bst: No hyphenation pattern has been}%
\typeout{** loaded for the language `#1'. Using the pattern for}%
\typeout{** the default language instead.}%
\else
\language=\csname l@#1\endcsname
\fi
#2}}
\providecommand{\BIBdecl}{\relax}
\BIBdecl

\bibitem{Dimakis11}
A.~G. Dimakis, K.~Ramchandran, Y.~Wu, and C.~Suh, ``A survey on network codes
  for distributed storage,'' \emph{Proceedings of the IEEE}, vol.~99, no.~3,
  pp. 476--489, 2011.

\bibitem{Balaji18}
S.~Balaji, M.~N. Krishnan, M.~Vajha, V.~Ramkumar, B.~Sasidharan, and P.~V.
  Kumar, ``Erasure coding for distributed storage: An overview,'' \emph{Science
  China Information Sciences}, vol.~61, no.~10, p. 100301, 2018.

\bibitem{Ghemawat03}
S.~Ghemawat, H.~Gobioff, and S.-T. Leung, ``The google file system,'' in
  \emph{Proceedings of the nineteenth ACM symposium on Operating systems
  principles}, 2003, pp. 29--43.

\bibitem{Shvachko10}
K.~Shvachko, H.~Kuang, S.~Radia, and R.~Chansler, ``The hadoop distributed file
  system,'' in \emph{2010 IEEE 26th symposium on mass storage systems and
  technologies (MSST)}.\hskip 1em plus 0.5em minus 0.4em\relax Ieee, 2010, pp.
  1--10.

\bibitem{Aktas20}
M.~Aktas, G.~Joshi, S.~Kadhe, F.~Kazemi, and E.~Soljanin, ``Service rate
  region: A new aspect of coded distributed system design,'' \emph{arXiv
  preprint arXiv:2009.01598}, 2020.

\bibitem{Dimakis10}
A.~G. Dimakis, P.~B. Godfrey, Y.~Wu, M.~J. Wainwright, and K.~Ramchandran,
  ``Network coding for distributed storage systems,'' \emph{IEEE transactions
  on information theory}, vol.~56, no.~9, pp. 4539--4551, 2010.

\bibitem{Rashmi09}
K.~Rashmi, N.~B. Shah, P.~V. Kumar, and K.~Ramchandran, ``Explicit construction
  of optimal exact regenerating codes for distributed storage,'' in \emph{2009
  47th Annual Allerton Conference on Communication, Control, and Computing
  (Allerton)}.\hskip 1em plus 0.5em minus 0.4em\relax IEEE, 2009, pp.
  1243--1249.

\bibitem{Rashmi11}
K.~V. Rashmi, N.~B. Shah, and P.~V. Kumar, ``Optimal exact-regenerating codes
  for distributed storage at the msr and mbr points via a product-matrix
  construction,'' \emph{IEEE Transactions on Information Theory}, vol.~57,
  no.~8, pp. 5227--5239, 2011.

\bibitem{pp14}
D.~S. Papailiopoulos and A.~G. Dimakis, ``Locally repairable codes,''
  \emph{IEEE Transactions on Information Theory}, vol.~60, no.~10, pp.
  5843--5855, 2014.

\bibitem{Huang2012}
C.~Huang, H.~Simitci, Y.~Xu, A.~Ogus, B.~Calder, P.~Gopalan, J.~Li, and
  S.~Yekhanin, ``Erasure coding in windows azure storage,'' in \emph{Presented
  as part of the 2012 $\{$USENIX$\}$ Annual Technical Conference
  ($\{$USENIX$\}$$\{$ATC$\}$ 12)}, 2012, pp. 15--26.

\bibitem{Shah16}
N.~B. Shah, K.~Lee, and K.~Ramchandran, ``When do redundant requests reduce
  latency?'' \emph{IEEE Transactions on Communications}, vol.~64, no.~2, pp.
  715--722, 2016.

\bibitem{Joshi12}
G.~Joshi, Y.~Liu, and E.~Soljanin, ``Coding for fast content download,'' in
  \emph{Communication, Control, and Computing (Allerton), 2012 50th Annual
  Allerton Conference on}.\hskip 1em plus 0.5em minus 0.4em\relax IEEE, 2012,
  pp. 326--333.

\bibitem{Joshi15}
G.~Joshi, E.~Soljanin, and G.~Wornell, ``Queues with redundancy: Latency-cost
  analysis,'' \emph{ACM SIGMETRICS Performance Evaluation Review}, vol.~43,
  no.~2, pp. 54--56, 2015.

\bibitem{Li18}
B.~Li, A.~Ramamoorthy, and R.~Srikant, ``Mean-field analysis of coding versus
  replication in large data storage systems,'' \emph{ACM Transactions on
  Modeling and Performance Evaluation of Computing Systems (TOMPECS)}, vol.~3,
  no.~1, p.~3, 2018.

\bibitem{Aggarwal17}
V.~Aggarwal, J.~Fan, and T.~Lan, ``Taming tail latency for erasure-coded,
  distributee storage systems,'' in \emph{IEEE INFOCOM 2017-IEEE Conference on
  Computer Communications}.\hskip 1em plus 0.5em minus 0.4em\relax IEEE, 2017,
  pp. 1--9.

\bibitem{Al2019}
A.~O. Al-Abbasi, V.~Aggarwal, and T.~Lan, ``Ttloc: Taming tail latency for
  erasure-coded cloud storage systems,'' \emph{IEEE Transactions on Network and
  Service Management}, vol.~16, no.~4, pp. 1609--1623, 2019.

\bibitem{Badita19}
A.~Badita, P.~Parag, and J.-F. Chamberland, ``Latency analysis for distributed
  coded storage systems,'' \emph{IEEE Transactions on Information Theory},
  2019.

\bibitem{Ayesta18}
U.~Ayesta, T.~Bodas, and I.~M. Verloop, ``On a unifying product form framework
  for redundancy models,'' \emph{Performance Evaluation}, vol. 127, pp.
  93--119, 2018.

\bibitem{Ayesta19}
------, ``On redundancy-d with cancel-on-start aka join-shortest-work (d),''
  \emph{ACM SIGMETRICS Performance Evaluation Review}, vol.~46, no.~2, pp.
  24--26, 2019.

\bibitem{Ayesta19b}
U.~Ayesta, T.~Bodas, J.~Dorsman, and I.~Verloop, ``A token-based central queue
  with order-independent service rates,'' \emph{arXiv preprint
  arXiv:1902.02137}, 2019.

\bibitem{Hellemans18}
T.~Hellemans and B.~Van~Houdt, ``On the power-of-d-choices with least loaded
  server selection,'' \emph{Proceedings of the ACM on Measurement and Analysis
  of Computing Systems}, vol.~2, no.~2, p.~27, 2018.

\bibitem{Mitzenmacher01}
M.~Mitzenmacher, ``The power of two choices in randomized load balancing,''
  \emph{IEEE Transactions on Parallel and Distributed Systems}, vol.~12,
  no.~10, pp. 1094--1104, 2001.

\bibitem{Bramson10}
M.~Bramson, Y.~Lu, and B.~Prabhakar, ``Randomized load balancing with general
  service time distributions,'' in \emph{ACM SIGMETRICS Performance Evaluation
  Review}, vol.~38, no.~1.\hskip 1em plus 0.5em minus 0.4em\relax ACM, 2010,
  pp. 275--286.

\bibitem{Mukherjee16}
D.~Mukherjee, S.~C. Borst, J.~S. van Leeuwaarden, and P.~A. Whiting,
  ``Universality of power-of-$ d $ load balancing in many-server systems,''
  \emph{arXiv preprint arXiv:1612.00723}, 2016.

\bibitem{Jinan20}
R.~Jinan, A.~Badita, T.~Bodas, and P.~Parag, ``Load balancing policies with
  server-side cancellation of replicas,'' \emph{arXiv preprint
  arXiv:2010.13575}, 2020.

\bibitem{Bramson11}
M.~Bramson \emph{et~al.}, ``Stability of join the shortest queue networks,''
  \emph{The Annals of Applied Probability}, vol.~21, no.~4, pp. 1568--1625,
  2011.

\bibitem{Bramson12}
M.~Bramson, Y.~Lu, and B.~Prabhakar, ``Asymptotic independence of queues under
  randomized load balancing,'' \emph{Queueing Systems}, vol.~71, no.~3, pp.
  247--292, 2012.

\bibitem{Hellemans19}
T.~Hellemans, T.~Bodas, and B.~Van~Houdt, ``Performance analysis of workload
  dependent load balancing policies,'' \emph{Proceedings of the ACM on
  Measurement and Analysis of Computing Systems}, vol.~3, no.~2, pp. 1--35,
  2019.

\bibitem{Shneer20}
S.~Shneer and A.~Stolyar, ``Large-scale parallel server system with
  multi-component jobs,'' \emph{arXiv preprint arXiv:2006.11256}, 2020.

\bibitem{Shah17}
V.~Shah, A.~Bouillard, and F.~Baccelli, ``Delay comparison of delivery and
  coding policies in data clusters,'' in \emph{Communication, Control, and
  Computing (Allerton), 2017 55th Annual Allerton Conference on}.\hskip 1em
  plus 0.5em minus 0.4em\relax IEEE, 2017, pp. 397--404.

\bibitem{Peng20}
P.~Peng, E.~Soljanin, and P.~Whiting, ``Diversity/parallelism trade-off in
  distributed systems with redundancy,'' \emph{arXiv preprint
  arXiv:2010.02147}, 2020.

\end{thebibliography}

\end{document}